\documentclass[conference]{IEEEtran}
\IEEEoverridecommandlockouts

\usepackage{cite}
\usepackage{amsmath,amssymb,amsfonts}
\usepackage{amsthm}
\usepackage{algorithmic}
\usepackage{graphicx}
\usepackage{textcomp}
\usepackage{xcolor}
\usepackage{braket}
\usepackage{subfigure}
\usepackage{todonotes}
\def\BibTeX{{\rm B\kern-.05em{\sc i\kern-.025em b}\kern-.08em
    T\kern-.1667em\lower.7ex\hbox{E}\kern-.125emX}}

\newtheorem{definition}{Definition}[section]
\newtheorem{lemma}{Lemma}[section]
\newtheorem{corollary}{Corollary}[section]

\begin{document}

\title{Normalized Gradient Descent for \\ 
Variational Quantum Algorithms\\
}
\author{\IEEEauthorblockN{Yudai Suzuki\IEEEauthorrefmark{1}, Hiroshi Yano\IEEEauthorrefmark{2}, Rudy Raymond\IEEEauthorrefmark{3}\IEEEauthorrefmark{4}, and Naoki Yamamoto\IEEEauthorrefmark{2}\IEEEauthorrefmark{4}}
\IEEEauthorblockA{\IEEEauthorrefmark{1}
Department of Mechanical Engineering,
Keio University, Hiyoshi 3-14-1, Kohoku, Yokohama 223-8522, Japan}
\IEEEauthorblockA{\IEEEauthorrefmark{2}
Department of Applied Physics and Physico-Informatics,
Keio University, Hiyoshi 3-14-1, Kohoku, Yokohama 223-8522, Japan}
\IEEEauthorblockA{\IEEEauthorrefmark{3}
IBM Quantum, IBM Japan, 19-21 Nihonbashi Chuo-ku, Tokyo 103-8510, Japan}
\IEEEauthorblockA{\IEEEauthorrefmark{4}
Quantum Computing Center, Keio University, Hiyoshi 3-14-1, Kohoku, Yokohama 223-8522, Japan}

}

\maketitle

\begin{abstract}
Variational quantum algorithms (VQAs) are promising methods that leverage noisy quantum computers and classical computing techniques for practical applications. 
In VQAs, the classical optimizers such as gradient-based optimizers are utilized to adjust the parameters of the quantum circuit so that the objective function is minimized. 
However, they often suffer from the so-called vanishing gradient or barren plateau issue. 
On the other hand, the normalized gradient descent (NGD) method, which employs the normalized gradient vector to update the parameters, has been successfully utilized in several optimization problems. 
Here, we study the performance of the NGD methods in the optimization of VQAs for the first time. 
Our goal is two-fold. 
The first is to examine the effectiveness of NGD and its variants for overcoming the vanishing gradient problems. 
The second is to propose a new NGD that can attain the faster convergence than the ordinary NGD. 
We performed numerical simulations of these gradient-based optimizers in the context of quantum chemistry where VQAs are used to find the ground state of a given Hamiltonian. 
The results show the effective convergence property of the NGD methods in VQAs, compared to the relevant optimizers without normalization. 
Moreover, we make use of some normalized gradient vectors at the past iteration steps to propose the novel \textit{historical NGD} that has a theoretical guarantee to accelerate the convergence speed, which is observed in the numerical experiments as well.
\end{abstract}


\section{Introduction}
Along with the recent rapid advances in quantum information processing devices, the increasing 
attention has been paid to the possibility that quantum computers can outperform the classical 
(conventional) computers in various research fields such as machine learning, chemistry, and finance. 
However, since the currently available quantum computers are not fault-tolerant, their capabilities 
are limited \cite{Preskill2018}. 
This has led the necessity to develop hybrid quantum-classical approaches that enable those noisy quantum devices to work, with the help of classical computers. 
The variational quantum algorithm (VQA) is one such strategy \cite{cerezo2020variational,bharti2021noisy,endo2021hybrid}. 
The VQA runs a parameterized quantum circuit (PQC) on a quantum computer, with variationaly 
updating the parameters by a classical optimizer to find a global minimum of the objective function. 
This approach is expected to show some quantum advantages, because PQCs with 
even short depth potentially have bigger expressibility than classical models 
such as the neural network. 
To date, a variety of VQAs has been proposed; the variational quantum eigensolver (VQE) for quantum chemistry 
\cite{peruzzo2014variational,mcclean2016theory,kandala2017hardware}, the 
quantum approximate optimization algorithm (QAOA) for combinatorial 
optimizations \cite{farhi2014quantum,farhi2016quantum,zhou2020quantum}, and 
quantum circuit learning algorithms for machine learning problems 
\cite{Havl_ek_2019,schuld2018circuitcentric,Mitarai_2018}.

Of course the performance of VQAs heavily depends on the power of classical optimizer, particularly the convergence speed of the optimization. 
Actually various optimizers have been tested in VQAs; 
the gradient-based optimizers such as adaptive moment estimation (ADAM) \cite{kingma2014adam}, the conjugate 
gradient (CG) \cite{hestenes1952methods}, the simultaneous perturbation stochastic approximation (SPSA) \cite{spall1992multivariate} and the 
natural gradient \cite{amari1998natural}, as well as the gradient-free optimizers such as Nelder-Mead \cite{nelder1965simplex} and 
COBYLA \cite{powell1994direct}. 
In this work, we study the gradient-based approach, with special attention to the serious issue recently recognized in this method. 
That is, it has been demonstrated that the VQA has the so-called vanishing gradient or the barren 
plateau issue, where the gradient vector of the objective function becomes exponentially small with the increase of the number of qubits \cite{mcclean2018barren,cerezo2021cost}. 
This makes the gradient-based optimizer inefficient for the optimization of VQAs. 
So far, several circumventing approaches have been proposed to remedy this issue, such as the 
initialization-engineering technique and the tailored PQCs \cite{grant2019initialization,verdon2019learning,hadfield2019quantum,yamamoto2019natural,stokes2020quantum}. 

Note that the similar vanishing gradient issue can generally occur for non-convex optimization 
problems. 
Additionally, in such non-convex optimization problems, the so-called exploding gradient issue 
is also observed, meaning that the norm of the gradient vector takes a huge value and as a result 
the training becomes unstable. 
These detrimental issues, however, can be circumvented via rather a simple method that uses the 
normalized gradient vector to update the parameter through the learning process. 
This method is called the {\it normalized gradient descent (NGD)}~\cite{hazan2015beyond}. 
Though the vanishing or exploding gradient issues never happen, the NGD could be disadvantageous 
in view of the convergence properties because it does not have a norm-dependent flexibility to 
search the optimal parameter. 
Nonetheless, under some conditions, the NGD is proven to evade the saddle points faster than the 
ordinary gradient descent method in the setting of continuous-time dynamics \cite{murray2019revisiting}.

Here, we study the performance of the NGD method in the optimization of VQAs.
To the best of our knowledge, this work is the first to investigate the NGD methods in regards to the VQAs.
This paper focuses on two objectives.
The first objective is to examine the effectiveness of NGD and its variants for resolving the vanishing gradient issues. 
We applied those optimizers to several VQE problems, where the ground state of a given Hamiltonian is variationally sought. 
The numerical simulations show the good convergence property of the NGD methods, in comparison with the relevant optimizers that do not use the normalized gradient vector.
As the second objective, we exploit the normalized gradient vectors at the past iteration steps to propose the \textit{historical NGD}. 
In particular, we derive a set of proper learning rates of the historical normalized gradient vectors, under the assumption of the strictly-locally-quasi-convex (SLQC) objective functions. 
Hence the historical NGD is guaranteed to show the faster convergence than the ordinary NGD, which is also demonstrated in the numerical experiments.

The rest of the paper is organized as follows. 
In Section~\ref{sec:preliminaries}, we show some gradient-based optimizers studied in this work. 
We then describe our historical NGD method that utilizes the normalized gradient vectors at the previous iteration step.
Subsequently, Section~\ref{sec:exp_result} demonstrates the numerical simulation of these optimizers for several VQE problems. 
At last, we conclude the paper in Section~\ref{sec:conclusion}.

\section{Preliminaries} \label{sec:preliminaries}
The gradient-based optimizers use the gradient vector of a objective function to update the 
parameters, for finding the global minimum of the function. 
The straightforward method is the gradient descent (GD), which is expressed as follows:
\begin{equation}
    {\bf x}_{t+1} = {\bf x}_{t} - \eta g_t,
\end{equation}
where ${\bf x}_{t}$ represents the set of parameters, $g_t=\nabla f({\bf x}_{t})$ is the gradient 
vector of the differentiable function $f({\bf x}_{t})$, and $\eta\in\mathbb{R}$ is the learning rate. 
GD can effectively find the global minimum when the objective function is convex. 
However, when the objective function is non-convex, GD often poorly perform owing to the trainability 
problems, typically the vanishing gradient issue caused by the plateau landscape of the objective 
function and the exploding gradient one due to the steep cliffs. 

To date, numerous variants of GD have been proposed to circumvent those drawbacks. 
In this section, we briefly describe some gradient-based optimizers used in this paper.

\subsection{Normalized Gradient Descent}
NGD is a simple method that resolves the aforementioned vanishing or exploding gradient issues 
by normalizing the gradient vector. 
NGD updates the parameters according to the following formula:
\begin{equation} \label{eq:ngd}
    {\bf x}_{t+1} = {\bf x}_{t} - \eta \hat{g}_t,
\end{equation}
where $\hat{g}_t=\nabla f({\bf x}_{t})/\|\nabla f({\bf x}_{t})\|$ is the normalized gradient vector. 
Because $\| \hat{g}_t \|=1$ for all $t$, NGD of course neither vanishes nor explodes. 
Hence we can expect NGD will show good convergence property as well as stable learning. 
Actually, it has been proven in the framework of continuous-time dynamics that NGD can escape 
the saddle point faster than GD \cite{murray2019revisiting}. 
Also, Ref.~\cite{hazan2015beyond} proved that NGD can converge to a global minimum for a wider class of functions, assuming the strictly-locally-quasi-convex (SLQC) property of the objective function, which we will explain later on. 

\subsection{Nesterov's Accelerated Gradient Method}
Nesterov's Accelerated Gradient (NAG) method \cite{nesterov1983method} is a simple modification of the momentum method\cite{polyak1964some}, 
which is also a variant of GD. 
In the momentum method, a moving average of the past gradients are taken into account to realize 
faster convergence and alleviate the oscillation along the ridges of the canyon in the landscape of 
the objective function.
The momentum method updates the parameters in the following way; 
\begin{align*}
    m_{t} = \beta m_{t-1} - \eta \nabla f({\bf x}_{t}), ~~
   {\bf x}_{t+1} = {\bf x}_{t} + m_{t}, 
\end{align*}
where $\beta$ and $\eta$ are positive scalars. 

As for NAG method, the update rule is represented as
\begin{align}
    {\bf x}_{t+1} &= {\bf y}_{t} - \eta \nabla f({\bf y}_{t}),\label{eq:upd_NSG}\\
    {\bf y}_{t} &= {\bf x}_{t} + \gamma_t ({\bf x}_{t}-{\bf x}_{t-1}),\label{eq:mom_NSG} \\
    \gamma_t &= \frac{\rho_{t-1}-1}{\rho_{t}}, \label{eq:coef1_NSG}\\
    \rho_{t} &= \frac{1+\sqrt{1+4\rho_{t-1}^{2}}}{2}.
\label{eq:coef2_NSG}
\end{align}
Here, $\rho_t$ and $\gamma_t$ at each iteration step are recursively calculated so that the optimal convergence rate for the smooth convex function is achieved \cite{nesterov1983method}.
In our experiment, similar to~\cite{sutskever2013importance} we set the initial value of $\rho$ as 1, i.e. $\rho_0 = 1$.

Note that the NAG method is rewritten as
\begin{align*}
    m_{t} = \beta m_{t-1} - \eta \nabla f({\bf x}_{t-1}+\beta m_{t-1}), ~~~
   {\bf x}_{t+1} = {\bf x}_{t} + m_{t}, 
\end{align*}
which shows that the only difference between the momentum method and NAG is the point at which 
we calculate the gradient \cite{sutskever2013importance}. 
The NAG method recently attracts much attention in the convex optimization community due to its good convergence property in some situations. 

\subsection{Adaptive Moment Estimation}
ADAM \cite{kingma2014adam} is a widely-used optimizer in the field of machine learning, particularly for deep neural networks. 
In a broad sense, ADAM utilizes the advantages of both the momentum method and RMSprops\cite{tieleman2012lecture}; 
ADAM not only keeps the moving average of the past gradients, but also computes the adaptive 
learning rate to reduce the oscillation. 
The update rule of ADAM is described as follows;
\begin{align}
   \bar{m}_{t} &= \beta_1 \bar{m}_{t-1} - (1-\beta_1)\nabla f({\bf x}_{t}),\\
   \bar{v}_{t} &= \beta_2 \bar{v}_{t-1} - (1-\beta_2)\nabla f({\bf x}_{t})^2,\\
   m_t &= \frac{\bar{m}_{t}}{1-\beta_1^{t+1}},\\
   \quad v_t &= \frac{\bar{v}_{t}}{1-\beta_2^{t+1}}, \\
    {\bf x}_{t+1} &= {\bf x}_{t} - \eta \frac{m_{t}}{\sqrt{v_t}+\epsilon} ,
\end{align}
where $\epsilon$ is a small constant introduced for numerical stability.
In \cite{kingma2014adam}, the hyperparameters are set as $\beta_1=0.9$, $\beta_2=0.999$ and $\epsilon=10^{-8}$, which we also choose in our experiments.

Note that we can take another type of gradient-based optimizer, which utilizes higher-order derivative 
information such as the Hessian of an objective function. 
Typically used are Newton's method and quasi Newton method (e.g., the 
Broyden–Fletcher–Goldfarb–Shanno (BFGS) algorithm \cite{fletcher2013practical} and a sequential least squares programming 
(SLSQP) algorithm \cite{kraft1994algorithm}). 
While they are computationally expensive compared to the variants of GD, these methods show 
faster convergence. 
However, this work focuses on the effectiveness of the optimizers with the normalized gradient vectors for VQAs, and thus we will not consider the 
higher-order gradient-based optimizer.


\section{Historical NGD}
\label{sec:method}

As mentioned before, the possible drawback of NGD is that it could be slower compared to GD 
due to the restricted norm condition. 
The historical NGD described here may be used to mitigate this issue; note that this method can 
also be applied to general optimization problem other than VQAs. 
The basics of this method is the provable convergence property of NGD under the assumption of the strictly-locally-quasi-convex (SLQC) objective function \cite{hazan2015beyond}. 
A particularly useful result is that, according to \cite{hazan2015beyond}, NGD can converge to an 
$\epsilon$-optimal minimum of the SLQC objective function with the rate $\mathcal{O}(1/\epsilon^2)$; 
in fact we will make use of the lemmas related to this fact, to derive the proper learning rates of NGD 
with historical gradients.

In this section, we firstly introduce the relationship between the SLQC and NGD shown in 
\cite{hazan2015beyond}, which is followed by deriving lemmas used to prove our result. 
Then we show that the one-step historical NGD can indeed accelerate the convergence speed. 
This method is further generalized to the historical NGD based on arbitrary $m$ normalized gradient 
vectors used in the past iteration steps.

\subsection{NGD for the Strictly-Locally-Quasi-Convex Function}

In a broad sense, the SLQC function is the generalization of unimodal (or quasi-convex) functions 
with multi-dimensions, which can take even a plateau shape. 
The definition of the SLQC is as follows \cite{hazan2015beyond}. 

\begin{definition}\textbf{(Local-Quasi-Convexity)}
\label{def:LocalQC}
Let ${\bf x},{\bf z}\in \mathbb{R}^d$ with the dimension $d$. 
Also let $\kappa$ and $\epsilon$ be positive constants. 
We consider a differentiable function $f: \mathbb{R}^d \mapsto \mathbb{R}$. 
Then $f$ is called $(\epsilon,\kappa,{\bf z})$-Strictly-Locally-Quasi-Convex (SLQC) at ${\bf x}$, 
if at least one of the following conditions holds: 
\begin{enumerate}
\item $f({\bf x})- f({\bf z})\leq  \epsilon$. 
\item $\|\nabla f({\bf x})\|> 0$, and for every ${\bf y}\in\mathbb{B}({\bf z},\epsilon/\kappa)$ it holds 
that $\braket{\nabla f({\bf x}), {\bf y}-{\bf x}}\leq 0$, 
\end{enumerate}
where $\mathbb{B}({\bf x},r)$ denotes a ball of radius $r$ around ${\bf x}$. 
\end{definition}

In \cite{hazan2015beyond}, the authors proved that NGD converges to the global minimum of 
a SLQC objective function, by taking the learning rate $\eta=\epsilon/\kappa$ in Eq.~\eqref{eq:ngd}. 
Note that the SLQC condition is not too strict; some intriguing SLQC functions 
are studied in \cite{hazan2015beyond}, such as the generalized linear model with certain setups.

Now we consider a SLQC function $f({\bf x})$, and assume that it has a local or global minimum 
point ${\bf x}^*$. 
Also recall that $\hat{g}_t=\nabla f({\bf x}_t)/\|\nabla f({\bf x}_t)\|$ is a normalized gradient vector, 
and the NGD policy is given by ${\bf x}_{t+1} = {\bf x}_t - \eta\hat{g}_t$ with the learning rate $\eta=\epsilon/\kappa$. 
Then the following three lemmas hold \cite{hazan2015beyond}. 
Note that, the first lemma is proved in \cite{hazan2015beyond} and we obtain the remaining lemmas by following a similar proof.

\begin{lemma}\label{lem:grad1} 
$$
\braket{\hat{g}_t,{\bf x}_t - {\bf x}^*} \ge \epsilon/\kappa. 
$$
\end{lemma}

Using this lemma, Ref.~\cite{hazan2015beyond} showed that at every update of ${\bf x}_{t}$ in NGD, 
the distance between ${\bf x}_t$ and ${\bf x}^*$ is reduced by at least $(\epsilon/\kappa)^2$.

\begin{lemma} \label{lem:grad2}
Define $\delta = \braket{\hat{g}_{t+1}, \hat{g}_{t}}$. Then,
$$
\braket{\hat{g}_{t+1},{\bf x}_{t} - {\bf x}^*} \ge \left(\epsilon/\kappa\right)\left(1 + \delta\right).
$$
\end{lemma}
\begin{proof}
Using Lemma~\ref{lem:grad1}, we have
\begin{eqnarray*}
&& \braket{\hat{g}_{t+1},{\bf x}_{t+1} - {\bf x}^*} \ge (\epsilon/\kappa)\\
&& \Leftrightarrow \braket{\hat{g}_{t+1},{\bf x}_{t} - (\epsilon/\kappa)\hat{g}_t - {\bf x}^*} 
    \ge (\epsilon/\kappa)\\
&& \Leftrightarrow \braket{\hat{g}_{t+1},{\bf x}_{t} -{\bf x}^*} 
     - (\epsilon/\kappa)\braket{\hat{g}_{t+1},\hat{g}_t} \ge (\epsilon/\kappa)\\
&& \Leftrightarrow \braket{\hat{g}_{t+1},{\bf x}_{t} - {\bf x}^*} \ge (1 + \delta) (\epsilon/\kappa).
\end{eqnarray*}
\end{proof}

\begin{lemma} \label{lem:grad3}
Define $\delta_{a,b} = \braket{\hat{g}_{t+a},\hat{g}_{t+b}}$. 
Then, for a natural number $m\in\mathbb{N}$,
$$
  \braket{\hat{g}_{t+m},{\bf x}_{t} - {\bf x}^*} \ge 
      \left(\epsilon/\kappa\right)\left(1 + \sum_{i=0}^{m-1}\delta_{m,i}\right) .
$$
\end{lemma}
\begin{proof}
Using Lemma~\ref{lem:grad1}, we have
\begin{eqnarray*}
&& \braket{\hat{g}_{t+m},{\bf x}_{t+m} - {\bf x}^*} \ge (\epsilon/\kappa)\\
&& \Leftrightarrow 
      \braket{\hat{g}_{t+m},{\bf x}_{t} - (\epsilon/\kappa)\sum_{i=0}^{m-1}\hat{g}_{t+i} - {\bf x}^*} 
       \ge (\epsilon/\kappa)\\
&& \Leftrightarrow
       \braket{\hat{g}_{t+m},{\bf x}_{t} -{\bf x}^*} 
       - (\epsilon/\kappa)\sum_{i=0}^{m-1}\braket{\hat{g}_{t+m},\hat{g}_{t+i}} \ge (\epsilon/\kappa)\\
&& \Leftrightarrow
       \braket{\hat{g}_{t+m},{\bf x}_{t} - {\bf x}^*} 
       \ge \left(1 + \sum_{i=0}^{m-1}\delta_{m,i}\right) (\epsilon/\kappa).
\end{eqnarray*}
\end{proof}

\subsection{One step historical NGD}

Here we show that the historical NGD, which updates the variable using $\hat{g}_{t+1}$ in addition 
to $\hat{g}_{t}$, helps accelerating the convergent speed of NGD compared to the ordinary NGD. 
Note that the inner product $\delta_t = \braket{\hat{g}_t,\hat{g}_{t+1}}$ used in the following result 
satisfies $-1 < \delta_t  \le 1$ because $\hat{g}_t$ and $\hat{g}_{t+1}$ are normalized.

\begin{corollary}
Let $\delta_t = \braket{\hat{g}_t,\hat{g}_{t+1}}$, satisfying $-1 < \delta_t  \le 1$. 
Consider the following update policy: 
$$
{\bf x}_{t+2} = {\bf x}_{t} + \eta_1 \hat{g}_t + \eta_2 \hat{g}_{t+1}, 
$$ 
where the learning rate $\eta_1, \eta_2$ are determined as follows: 
\begin{itemize}
\item When $ -1 < \delta_t \le (\sqrt{5}-1)/2$, 
\[
      \eta_1 = - \left(\frac{\epsilon}{\kappa}\right)\frac{1-\delta_t - \delta_t^2}{1-\delta_t^2}, ~~
      \eta_2 = -\left(\frac{\epsilon}{\kappa}\right)\frac{1}{1-\delta_t^2}. 
\]
\item When $(\sqrt{5}-1)/2 < \delta_t \le1$, 
\[
      \eta_1 = 0, ~~ \eta_2 = - \left(\frac{\epsilon}{\kappa}\right)\left(1 + \delta_t\right). 
\]
\end{itemize}
Then it holds that 
$$
\|{\bf x}_{t+2} - {\bf x}^*\|^2 <  \|{\bf x}_{t} - {\bf x}^*\|^2  - c \left(\epsilon/\kappa\right)^2, 
$$
for $ c \ge 2$. 

\end{corollary}
\begin{proof}
By simple algebra, we can obtain
\begin{equation} \label{eq:dist_t+2_optimal}
\begin{split}
\|{\bf x}_{t+2} - {\bf x}^*\|^2 &= \braket{{\bf x}_{t+2} - {\bf x}^*, {\bf x}_{t+2}-{\bf x}^*} \\
&= \|{\bf x}_{t} - {\bf x}^*\|^2 + \| \eta_1\hat{g}_t + \eta_2 \hat{g}_{t+1} \|^2 \\ 
 &\quad +  2 \braket{{\bf x}_{t}-{\bf x}^*,\eta_1\hat{g}_t + \eta_2 \hat{g}_{t+1} } \\
&= \|{\bf x}_{t} - {\bf x}^*\|^2 + \eta_1^2 + \eta_2^2 + 2\eta_1\eta_2\delta_t\\ 
&\quad + 2\eta_1\braket{\hat{g}_t,{\bf x}_{t}-{\bf x}^*} + 2\eta_2\braket{\hat{g}_{t+1},{\bf x}_{t}-{\bf x}^*}\\
&\le \|{\bf x}_{t} - {\bf x}^*\|^2 + \eta_1^2 + \eta_2^2 + 2\eta_1\eta_2\delta_t \\
&\quad+ 2\eta_1\left(\epsilon/\kappa\right) + 2\eta_2(\epsilon/\kappa)(1+\delta_t).
\end{split}
\end{equation}
Here, we substitute ${\bf x}_{t} + \eta_1 \hat{g}_t + \eta_2 \hat{g}_{t+1}$ for ${\bf x}_{t+2}$ in the second equality.
Also, the last inequality is due to Lemmas \ref{lem:grad1} and \ref{lem:grad2} as well as 
$\eta_1, \eta_2 \le 0$. 
The proof follows by substituting the values of $\eta_1$ and $\eta_2$ to obtain:
\begin{itemize}
\item  When $ -1 < \delta_t \le (\sqrt{5}-1)/2$, 
$$
\|{\bf x}_{t+2} - {\bf x}^*\|^2  \le \|{\bf x}_{t} - {\bf x}^*\|^2 - \frac{2-\delta_t^2}{\left(1-\delta_t^2\right)}\left(\epsilon/\kappa\right)^2.
$$

\item When $(\sqrt{5}-1)/2 < \delta_t  \le 1 $, 
$$
\|{\bf x}_{t+2} - {\bf x}^*\|^2  \le \|{\bf x}_{t} - {\bf x}^*\|^2 - \left(1+\delta_t\right)^2 \left(\epsilon/\kappa\right)^2.
$$
\end{itemize}
\end{proof}

Note that, when the ordinary NGD is successively applied twice, we have 
\[
     \|{\bf x}_{t+2} - {\bf x}^*\|^2 \le  \|{\bf x}_{t} - {\bf x}^*\|^2  - 2\left(\epsilon/\kappa\right)^2, 
\]
meaning that the historical NGD can show better convergence than the ordinary NGD.

\subsection{Multi-step historical NGD}

In the previous subsection, we consider the historical NGD based on the gradient information 
at one past iteration step. 
Here we generalize the idea to NGD with arbitrary $m$ normalized gradient vectors in the past 
iteration steps.

Note that, in the previous case, the key to have the result is the proper choice of learning rates 
$\eta_1$ and $\eta_2$, and this can be done by quadratic programming. 
Actually, the learning rates can be obtained by minimizing the second and subsequent terms in 
the right hand side of the last inequality in Eq.~\eqref{eq:dist_t+2_optimal}, which can be expressed 
as the quadratic function (denoted as $h(\eta_1,\eta_2)$) in the following way;
\begin{equation} \label{eq:motivating_ex}
\begin{split}
h(\eta_1,\eta_2) &=  \eta_1^2 + \eta_2^2 + 2\eta_1\eta_2\delta_t + 2\eta_1\left(\epsilon/\kappa\right) + 2\eta_2(\epsilon/\kappa)(1+\delta_t)\\
&= \left[\begin{array}{cc} \eta_{1} & \eta_{2} \end{array} \right] \left[\begin{array}{cc} 1 & \delta_{t}  \\ \delta_{t} & 1  \end{array} \right] \left[\begin{array}{c} \eta_{1} \\ \eta_{2} \\ \end{array} \right] \\
&\quad + \left[\begin{array}{cc} 2\epsilon/\kappa & 2\epsilon/\kappa(1+\delta_{t}) \\ \end{array} \right] \left[\begin{array}{cc} \eta_{1} \\ \eta_{2} \\ \end{array} \right]\\
    &= {\bf y}^{\mathrm{T}} A {\bf y} + C{\bf y}.
\end{split}
\end{equation} 
The learning rates are obtained by solving the minimum of the above quadratic function, under 
the constraint $\eta_1, \eta_2 \le 0$. 
We can apply this strategy to the case using $m$ previous gradient vectors for the historical NGD, 
which we call NGD$m$. 
The update rule of the NGD$m$ is given by 
\begin{equation}
    {\bf x}_{t+m} = {\bf x}_{t} + \sum_{i=1}^{m} \eta_i \hat{g}_{t+i-1},
\end{equation}
where $\hat{g}_{t+m}$ is the normalized gradient vector at $t+m$ iteration step.
Then we determine the learning rates $\{\eta_k\}_{k=1}^{m}$ so that ${\bf x}_{t+m}$ is closer to 
${\bf x}^{*}$ than that via the ordinary NGD. 
More precisely, we derive the inequality connecting $\|{\bf x}_{t+m} - {\bf x}^*\|^2$ to 
$\|{\bf x}_{t} - {\bf x}^*\|^2$, like Eq.~\eqref{eq:dist_t+2_optimal}. 
Consequently, the problem for determining the proper learning rates reduces to the one to find 
the minimum of the following quadratic function $h({\bf y})$ with ${\bf y} = [\eta_1,\ldots,\eta_m]$:
\begin{equation}
\begin{split}
    h({\bf y}) &= {\bf y}^{\mathrm{T}} A {\bf y} + C{\bf y}, \\
    \quad A &=  
     \left[\begin{array}{cccc} 1 & \delta_{0,1} & \ldots & \delta_{0,m-1}  \\ \delta_{0,1} & 1 & \ldots & \delta_{1,m-1}  \\ \vdots & \vdots& \ddots & \vdots  \\ \delta_{0,m-1} & \delta_{1,m-1} & \ldots & 1 \end{array} \right],  \\
    C &= 2\epsilon/\kappa \left[\begin{array}{cccc} 1 & 1+\delta_{0,1} & \ldots & 1+\sum_{i=0}^{m-2}\delta_{i,m-1}\\ \end{array} \right],\\
\end{split}
\end{equation}
under the constraint ${\bf y} \le 0$, where $\delta_{i,j}=\braket{\hat{g}_{t+i},\hat{g}_{t+j}}$.
Here we utilize Lemma \ref{lem:grad3} to derive $C$ in $h({\bf y})$. 
Since the quadratic programming can be efficiently solved, our method can compute the proper 
learning rates at each iteration step quickly. 
However there is a caveat for this method in terms of numerical calculation; that is, the computation 
process can become unstable. 
For instance, when we use NGD2 (which is exactly the same as the one-step historical case), 
$\delta=-1$ results in the divergence of $\eta_1$ and $\eta_2$. 
This is because the minimization of the function $h(\eta_1,\eta_2)$ in Eq.~\eqref{eq:motivating_ex} 
is reduced to lowering $\eta_2$ as small as possible, as shown below; 
\begin{equation*}
\begin{split}
    h(\eta_1,\eta_2) &= \eta_{1}^{2} + \eta_{2}^{2} -2 \eta_{1}\eta_{2} + 2\eta_1\left(\epsilon/\kappa\right)\\
    &= \left(\eta_1-\eta_2+\epsilon/\kappa \right)^2 +2\eta_2(\epsilon/\kappa) - (\epsilon/\kappa)^2.
\end{split}
\end{equation*}
Hence, in the numerical simulation shown in the next section, we add extra constraint ${\bf y} \ge k$ 
with a negative constant $k$ to avoid the computational instability.

Lastly note that the proposed NGD is somehow similar to the momentum method in the sense 
that both use the past gradient information. 
However, our proposal differs with respect to the update rule, which is derived based on SLQC 
property where NGD can converge to the global minimum.


\section{Numerical simulations}\label{sec:exp_result}

In this section, we show numerical simulations for several optimization problems appearing quantum 
chemistry, to test the performance of NGD. 
The goal is to find the ground state of a given Hamiltonian $H$, using VQE. 
The basic procedure of VQE is as follows; given a PQC $U(\theta)$ with $n_p$ tunable parameters 
$\theta = \{\theta_1,\ldots,\theta_{n_p}\}$, $\theta$ is repeatedly updated by a classical optimizer 
so that the energy (the objective function) $f(\theta) = \braket{\Phi(\theta) |H|\Phi(\theta)}$ is 
reduced to its minimum, where $\ket{\Phi(\theta)} = U(\theta)\ket{{\bf 0}}$ with an initial state 
$\ket{{\bf 0}}$. 
We study five VQE problems; we begin with a toy problem and then move to four quantum 
chemistry problems studying ${\rm H}_2$, LiH, ${\rm H}_4$, and the transverse field Ising model.

All numerical simulations are performed using the statevector simulator which does not introduce 
a statistical error due to the measurement process, on Qiskit \cite{Qiskit} (version 0.24). 
As for the classical optimizers, we consider GD, NAG, ADAM (i.e., gradient-based optimizers 
without normalization), NGD, and the normalized NAG. 
Note that we use the parameter shift rule \cite{Mitarai_2018,schuld2019evaluating} to calculate 
the gradient vector of the objective functions. 
Also the learning rate of each optimizer is fixed to $0.05$. 
As for the historical NGD, we also set $\eta = \epsilon/\kappa=0.05$. 
To solve the quadratic programming problem discussed in Section III C, we use CVXOPT~\cite{2020ascl.soft08017A}, 
a package for convex optimization; in particular we choose the negative constant $k=-1000$, 
for the purpose of mitigating the computational instability.

\begin{figure*}[h]
\begin{center}
\begin{tabular}{ccc}
\subfigure[GD]{
\includegraphics[scale=0.33]{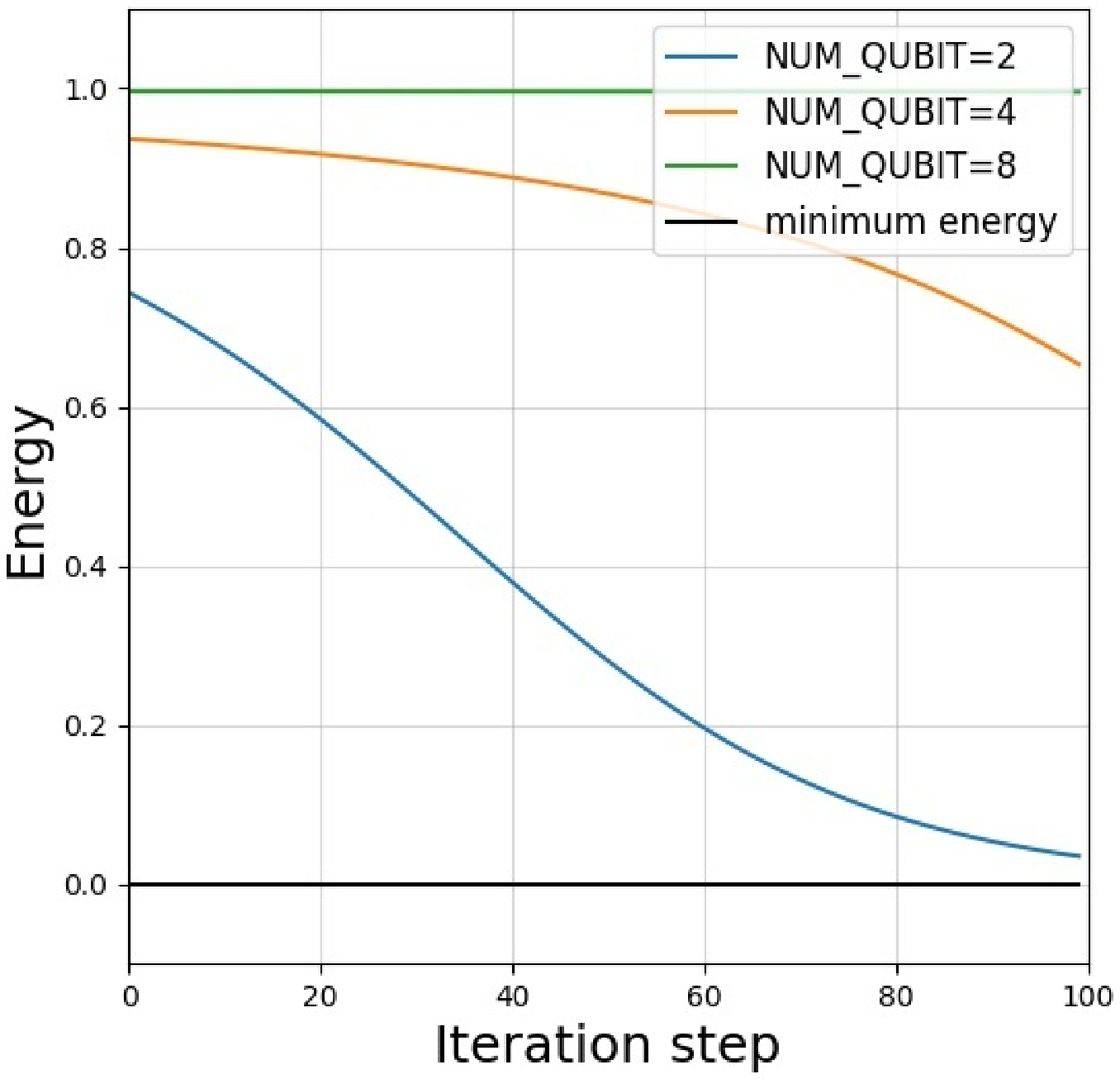}
\label{fig:wup_gd}
} &
\subfigure[NAG]{
\includegraphics[scale=0.33]{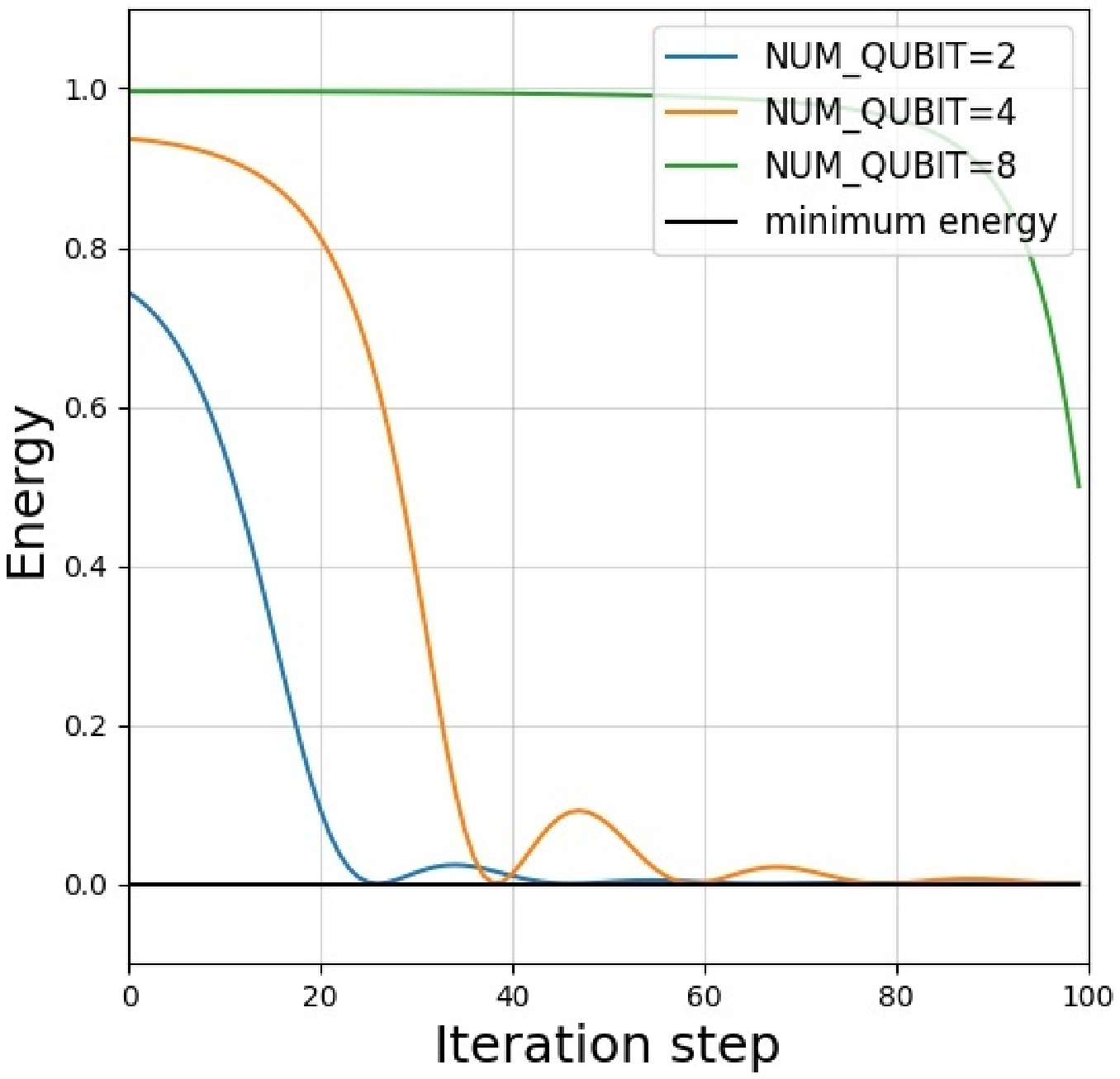}
\label{fig:wup_nag}
} &
\subfigure[ADAM]{
\includegraphics[scale=0.33]{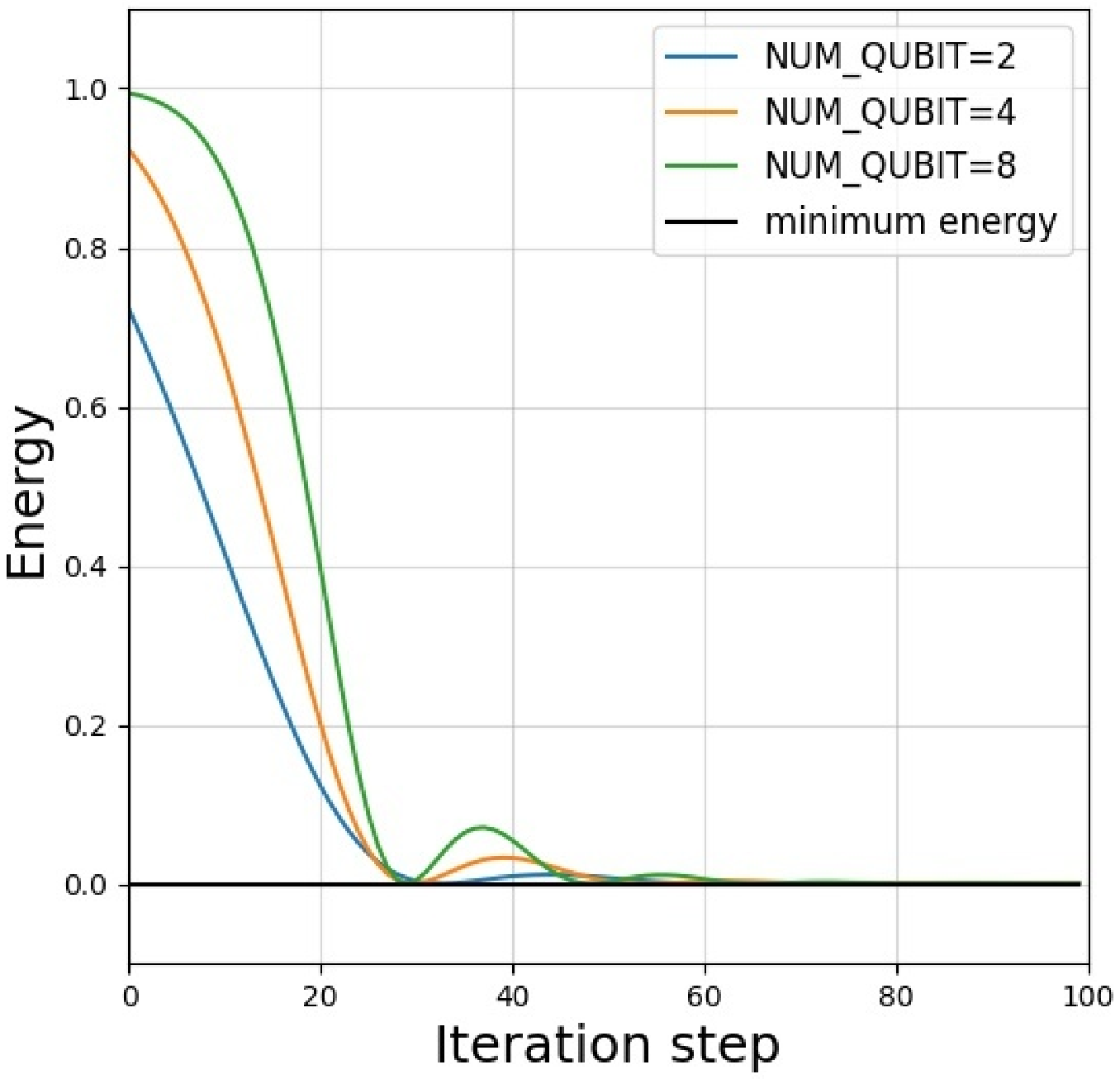}
\label{fig:wup_adam}
} \\

\subfigure[NGD]{
\includegraphics[scale=0.33]{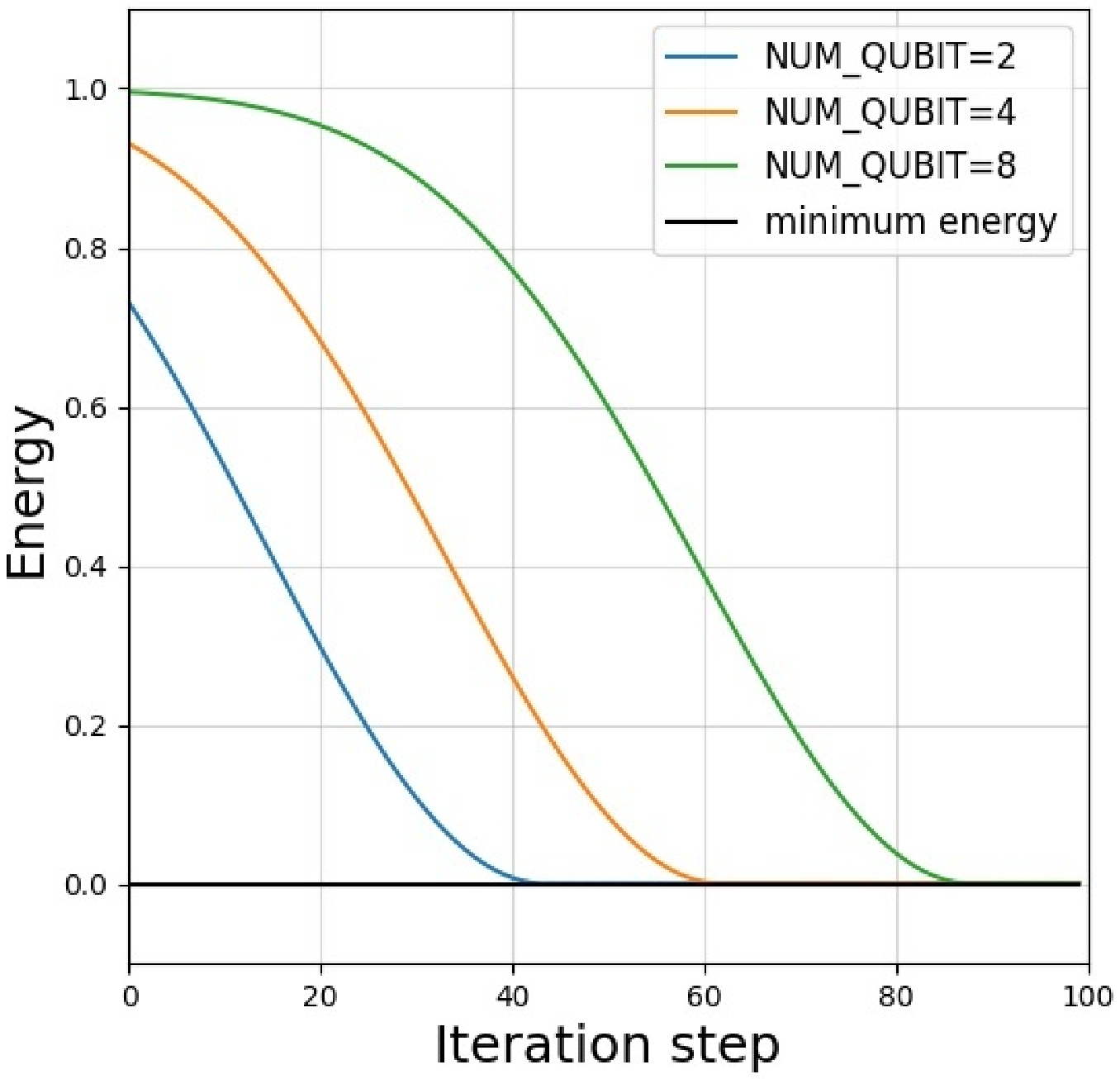}
\label{fig:wup_ngd}
} &
\subfigure[The normalized NAG]{
\includegraphics[scale=0.33]{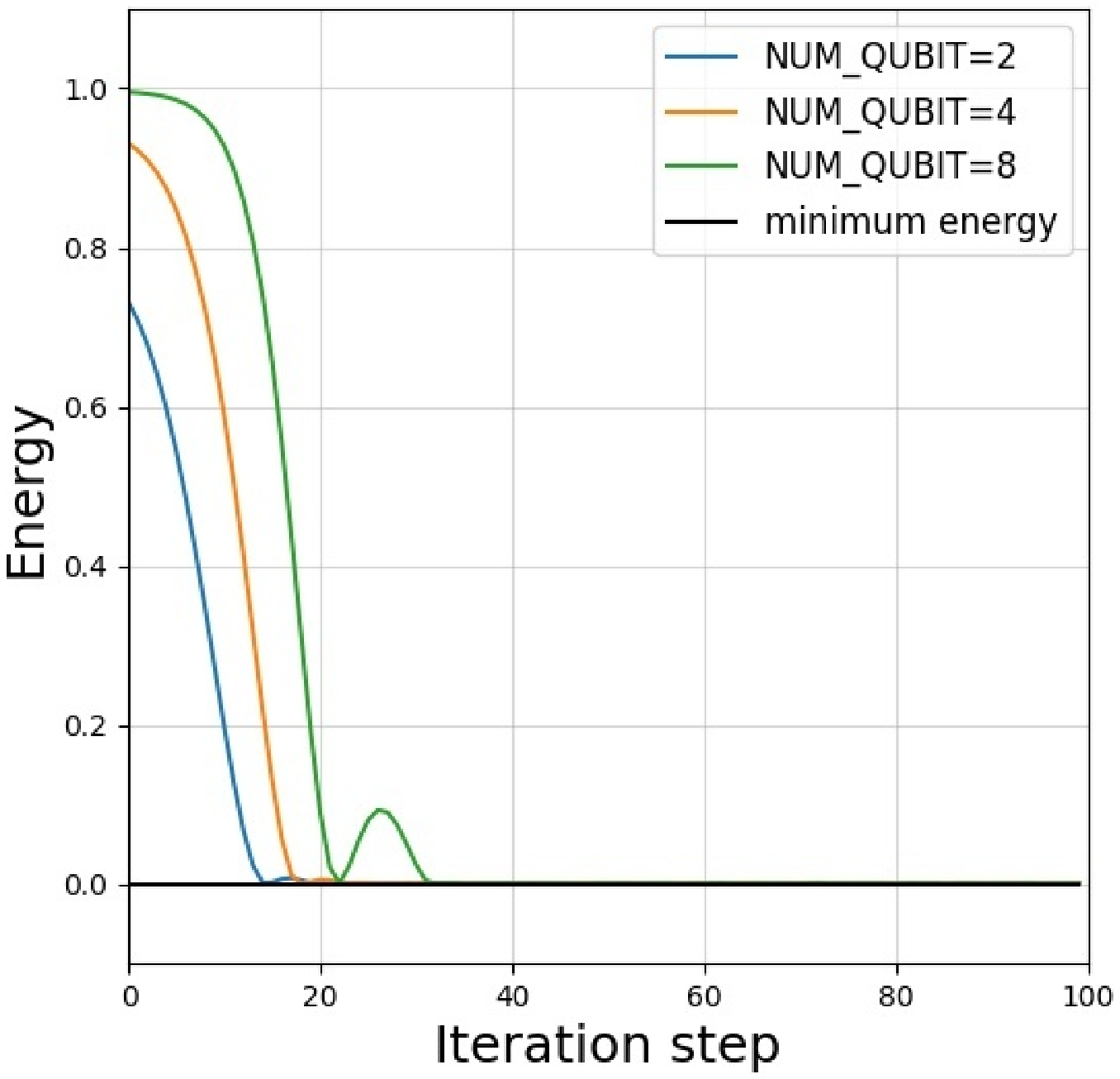}
\label{fig:wup_nnag}
} &
\subfigure[NGD2]{
\includegraphics[scale=0.33]{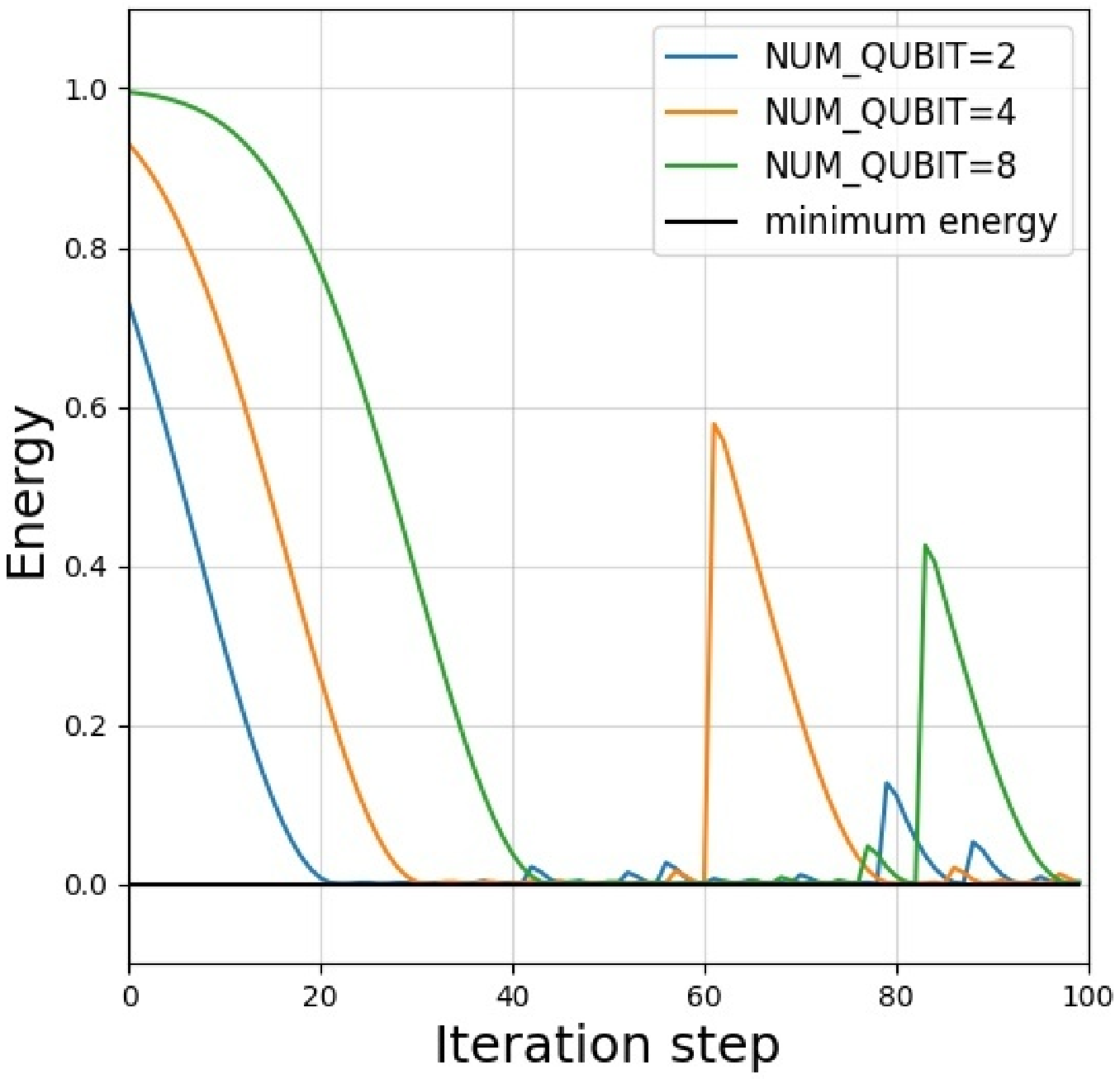}
\label{fig:wup_ngd2}
} \\
\end{tabular}
\caption{Energy of the narrow gorge potential in the case of $n=2$ (blue), $4$ (orange), 
and $8$ (green) against the iteration steps for each optimizer. 
The upper three panels (a, b, c) and lower three panels (d, e, f) show the results for the optimizers 
without normalized gradient vector and those for the optimizers with normalization, respectively. }
\label{fig:gorge_res}
\end{center}
\end{figure*}

\subsection{Toy narrow gorge problem}

We begin with a toy \textit{narrow gorge} problem, which was studied in \cite{arrasmith2021equivalence}. 
The narrow gorge is a type of the energy landscape, such that the well around the minimum shrinks 
as the number of qubits increases. 
As a result, the vanishing gradient issue can be well observed in this problem. 
That is, the norm of the gradient vector rapidly decreases in all the parameter space except at around 
the minimum, as the number of qubits increases. 
Hence we expect to see the informative difference by comparing the convergence speed of the 
optimizers with and without normalization of the gradient, depending on the number of qubits. 
We take the same problem setting as \cite{cerezo2021cost}. 
The Hamiltonian is $H=\sum_{k=1}^{n} \sigma_{X}^{(k)}$, whose ground state is 
$\ket{{\bf 0}}=\ket{0}^{\otimes n}$, where $\sigma_{X}$ is the Pauli $X$ operator and $n$ is the number 
of qubits. 
The PQC is chosen as $U(\theta)=\otimes_{k=1}^{n} e^{-i\theta_k\sigma_{X}^{(k)}}$. 
The goal is to optimize the parameters $\theta=\{\theta_k\}_{k=1}^{n}$ so that 
$U(\theta)\ket{{\bf 0}}=\ket{{\bf 0}}$.

In this work, we perform the simulation with different number of qubits $n=2, 4, 8$, to see that NGD indeed resolves the vanishing gradient issue and, at the same time, to evaluate the convergent speed of NGD. 
Note that the number of qubits is equal to that of parameters, due to the tensor-product structure 
of the aforementioned PQC. 
The optimal parameters of this task are all zeros, i.e. $\theta_k = 0$ for all $k$. 
Hence we set the initial parameters as $\theta_k=\pi/2$ for all $k$, to avoid that the initial point 
would immediately get close to the optimal point. 
Also the total number of iterations for the optimization is fixed to 100.

In Fig. \ref{fig:gorge_res} we show the energy of the narrow gorge potential, against the iteration 
steps for each optimizer, where the blue, orange, and green lines in each figure show the case of 
$n=2, 4$, and $8$, respectively. 
The upper three panels show the results for the optimizers without the normalized gradient vector, 
and the lower three show the optimizers with the normalized gradient vector. 
From these figures, we observe that the optimizers without normalization crucially suffer from 
the vanishing gradient issue, except for ADAM. 
On the other hand, the optimizers with normalization can still decrease the energy even when $n=8$. 
Note that the performance of optimizers appears to be degraded as $n$ increases, because the 
distance between the initial and the optimal points gets larger.

Moreover, Fig.~\ref{fig:gorge_res} (f) shows that the one-step historical NGD converges faster 
than the ordinary NGD, as proven in Corollary III.1. 
Here, we study if the convergence speed can be further improved by increasing the number 
of gradient information at the previous iteration steps. 
Fig.~\ref{fig:warmup_ngdm} shows the energy of narrow gorge potential for the case of $n=8$, 
for NGD, NGD2, NGD3, and NGD4. 
Recall that NGD$m$ means the historical NGD with $m$ past normalized gradient vectors. 
Clearly, Fig.~\ref{fig:warmup_ngdm} shows that NGD$m$ with bigger $m$ converges faster. 
However, this result also shows the issue of computational instability of NGD$m$, as indicated 
in Section \ref{sec:method}. 
In particular, this issue seems more likely to occur, as we utilize more normalized gradient vectors. 
Thus, there is a room for improvement in our method, to fully exploit the past normalized gradient 
vectors without suffering from this computational instability issue. 
Yet, we underscore that NGD4 (as well as the Normalized NAG) is the first to reach within an 
error of $10^{-2}$ from the minimum energy.

\begin{figure}[tbhp]
    \centering
    \includegraphics[width=\linewidth]{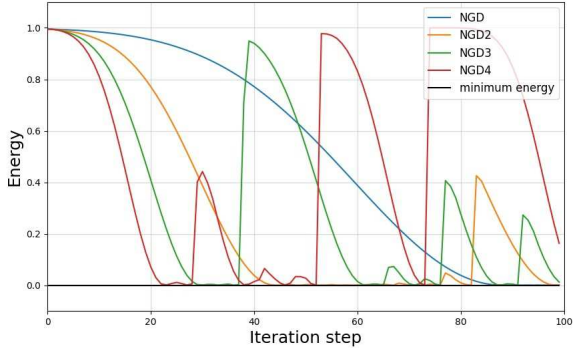}
    \caption{Energy of the narrow gorge potential in the case of $n=8$, against the iteration steps 
    for NGD, NGD2, NGD3, and NGD4.}
    \label{fig:warmup_ngdm}
\end{figure}

\subsection{$H_2$ molecule}

\begin{figure}[tbp]
    \centering
    \includegraphics[width=\linewidth]{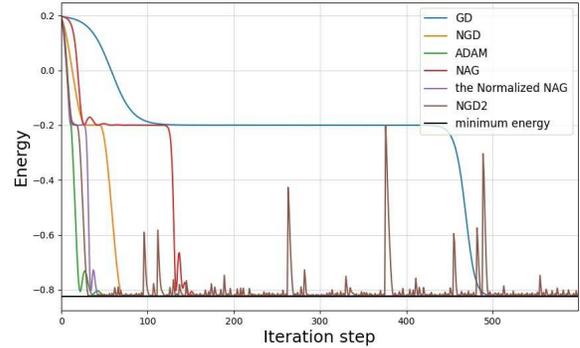}
    \caption{Energy of ${\rm H}_2$ molecule against the iteration steps for each optimizer}
    \label{fig:H2_res}
\end{figure}

We next examine the problem of finding the ground state of ${\rm H}_2$ molecule. 
The simplified Hamiltonian of ${\rm H}_2$ is expressed as
\begin{equation*}
    H=\alpha (\sigma_Z \otimes I + I \otimes \sigma_Z) + \beta (\sigma_X \otimes \sigma_X),
\end{equation*}
where $\sigma_Z$ represents Pauli $Z$ operator and $(\alpha,\beta) = (0.4,0.2)$. 
It was reported \cite{yamamoto2019natural} that the VQE using GD with learning rate $\eta=0.05$ 
gets stuck in a plateau for a while and then escapes later, when a single depth 
$Ry$ ansatz with initial parameters  $(\theta_1,\theta_2,\theta_3,\theta_4)=(7\pi/32,\pi/2,0,0)$ are used. 
This phenomenon arises because GD first arrives in the vicinity of the first excited state, where the 
gradient vector vanishes. 
Here we test GD, NAG, ADAM, NGD, the normalized NAG, and NGD2 with the same $R_y$ 
anzats, to see if they would be trapped in this plateau and, when trapped, how fast 
they can escape from it.

Figure \ref{fig:H2_res} shows that all optimizers get stuck at the first excited 
state with energy $-0.2$. 
But notably, the optimizers with the normalized gradient vector evade the plateau 
faster than the others without normalization. 
For example, NGD can get out of the first excited state around 50 iteration steps, 
while GD requires 450 iterations. 
Moreover, we can also see the normalized NAG method outperforms the ordinary NAG. 
Importantly, NGD2 is the first to reach the minimum, while ADAM can get out of the 
plateau the fastest.

\subsection{LiH molecule}

\begin{figure}[tbp]
    \centering
    \includegraphics[width=\linewidth]{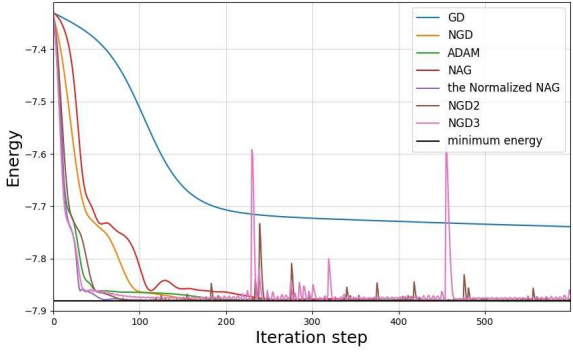}
    \caption{Energy of LiH molecule against the iteration steps for each optimizer}
    \label{fig:lih_res}
\end{figure}

The next case-study is on the LiH molecule; we consider the case where the 
interatomic distance is 1.5 \AA, resulting that two unoccupied orbitals are 
removed and the core is frozen. 
The Hamiltonian of LiH is constructed in the following way; the fermionic 
Hamiltonian is first constructed by the Hartree-Fock calculation with STO-3G 
basis \cite{hehre1969self} using the PySCF package \cite{sun2018pyscf}, which is then converted by the parity encoding method \cite{seeley2012bravyi} to the Hamiltonian. 
We apply VQE with two-depth $Ry$ ansatz, where the initial parameters are all zeros. 
The optimizers are GD, NAG, ADAM, NGD, the normalized NAG, NGD2, and NGD3.

The result is shown in Fig.~\ref{fig:lih_res}. 
In terms of the convergence speed, NGD and the normalized NAG are superior to GD 
and NAG, respectively. 
Also, NGD2 and the Normalized NAG reach the minimum the fastest, while NGD3 and 
ADAM are competitive with them at the beginning of the optimization process. 
Notably, NGD3 falls behind NGD2 to converge to the minimum, despite the fact 
that NGD3 utilizes more gradient information at the past iteration step.
The reason may be the computational instability that occurs when the learning rates 
are computed.
In fact, NGD3 shows the fastest convergence at first, but the energy begins to 
fluctuate at certain timestep.

\subsection{${\rm H}_4$ molecule}

\begin{figure}[tbp]
    \centering
    \includegraphics[width=\linewidth]{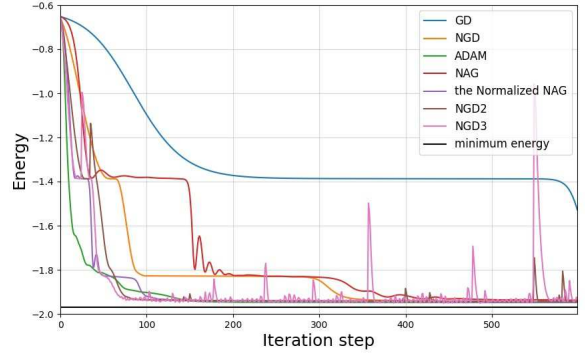}
    \caption{Energy of ${\rm H}_4$ molecule against the iteration steps for each optimizer}
    \label{fig:h4_res}
\end{figure}

We here consider the ${\rm H}_4$ molecule with square configuration with interatomic distance 1.277 \AA. 
The Hamiltonian of ${\rm H}_4$ is constructed in the same way to the LiH case. 
We use the five-depth $Ry$ ansatz as PQC, where the initial parameters are all zeros, 
and the same set of optimizers.

In this case, all optimizers cannot arrive at the minimum within 600 iterations. 
However, we still observe the better convergence property of the optimizers that use 
the normalized gradient vector. 
Namely, NGD and the normalized NAG method quickly converge to the local minimum 
about -1.95, in comparison with GD and NAG method, respectively.
Moreover, NGD3 reaches the lowest value the fastest. 
Note that the reason of not achieving the exact minimum may not be attributed to 
the optimizers, but other factors such as the insufficient expressibility of the 
PQC ansatz.

\subsection{Transverse field Ising model}

\begin{figure}[tbp]
    \centering
    \includegraphics[width=\linewidth]{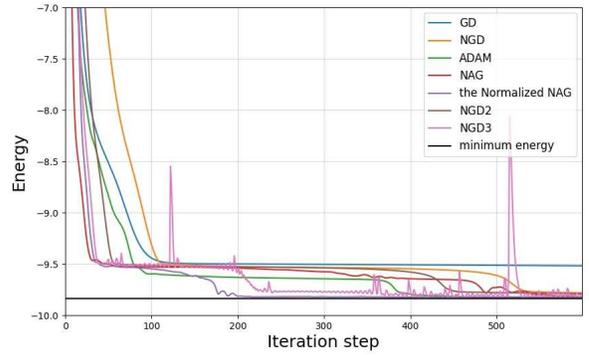}
    \caption{Energy of the transverse field Ising model, against the iteration 
    steps for each optimizer}
    \label{fig:Is_res}
\end{figure}

Lastly, we focus on the transverse field Ising model studied in 
\cite{sweke2020stochastic}, whose Hamiltonian is expressed as
\begin{equation*}
    H = \sum_{i=1}^{n-1} \sigma_{Z}^{(i)}\sigma_{Z}^{(i+1)} + \sum_{i=1}^{n} \sigma_{X}^{(i)}.
\end{equation*}
Here we set $n=8$. 
Again we use the two-depth Ry ansatz with all initial parameters given by $\pi/2$.

Unlike the previous problems, Fig.~\ref{fig:Is_res} shows the opposite results 
at the beginning of the optimization process; GD converges faster than NGD. 
This is because the chosen initial parameter is so far from the optimal point in 
the parameter space and the norm of the gradient vector is much bigger than one 
at around the initial point. 
This is actually seen in Fig.~\ref{fig:norm_grad}, showing the norm of the gradient 
vector against the optimization iterations for GD. 
However, this does not mean that NGD is inferior to GD.
In fact, Fig.~\ref{fig:Is_res} shows that GD gets stuck in a plateau, while NGD 
can escape from the plateau after roughly 500 iterations despite the poor convergence 
in the beginning.
This different behaviors can be, again, explained by Fig.~\ref{fig:norm_grad}, 
showing that the norm of the gradient vector quickly decreases as the optimization 
proceeds.
This reflects the effectiveness of the optimizers with normalized gradient vector 
for the convergence.

\begin{figure}
    \centering
    \includegraphics[width=\linewidth]{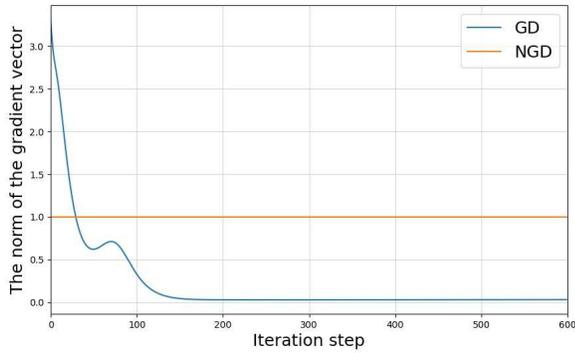}
    \caption{The norm of the gradient vectors against the iteration steps of GD for 
    the VQE calculation of the transverse field Ising model. For comparison, we also 
    show the norm of the gradient vector for NGD, which is always one.}
    \label{fig:norm_grad}
\end{figure}

\section{Conclusion}\label{sec:conclusion}

In this paper, we apply the NGD method to VQAs, to overcome the vanishing gradient issue. 
Several numerical simulations actually show that the optimizers with the normalized 
gradient vector have good convergence property, compared to the optimizers without 
normalization. 
Moreover, we proposed a new NGD that uses some normalized gradient vectors 
computed in the past optimization steps; this historical NGD is guaranteed to have 
a faster convergence property compared to the ordinary NGD, and actually we have 
demonstrated that this indeed accelerates the convergence speed of NGD.

We hope this work will pave a new way to deal with the vanishing gradient problem, 
that often appears in the optimization process of the VQAs.
Note also that, since the application of the historical NGD is not limited to 
VQAs, we hope that our method might be useful in e.g., machine learning.

\section*{Acknowledgment}
This work was supported by MEXT Quantum Leap Flagship Program Grant Number 
JPMXS0118067285 and JPMXS0120319794.
RR would like to thank Daisuke Okanohara for introducing~\cite{hazan2015beyond} that inspired him to conceive the idea of historical NGD. 
\bibliographystyle{IEEEtran}
\bibliography{references}

\end{document}